\documentclass[a4paper,UKenglish,cleveref, autoref, thm-restate,authorcolumns]{lipics-v2019}
\title{Simplifying Activity-on-Edge Graphs} 

\titlerunning{Simplifying Activity-on-Edge Graphs} %TODO optional, please use if title is longer than one line

\author{David Eppstein}{University of California, Irvine, United States }{eppstein@uci.edu}{}{}%TODO mandatory, please use full name; only 1 author per \author macro; first two parameters are mandatory, other parameters can be empty. Please provide at least the name of the affiliation and the country. The full address is optional

\author{Daniel Frishberg}{University of California, Irvine, United States }{dfrishbe@uci.edu}{https://orcid.org/0000-0002-1861-5439}{}%TODO mandatory, please use full name; only 1 author per \author macro; first two parameters are mandatory, other parameters can be empty. Please provide at least the name of the affiliation and the country. The full address is optional

\author{Elham Havvaei}{University of California, Irvine, United States }{ehavvaei@uci.edu}{https://orcid.org/0000-0003-0069-2863}{}%TODO mandatory, please use full name; only 1 author per \author macro; first two parameters are mandatory, other parameters can be empty. Please provide at least the name of the affiliation and the country. The full address is optional

\authorrunning{D. Eppstein et al.} %TODO mandatory. First: Use abbreviated first/middle names. Second (only in severe cases): Use first author plus 'et al.'

\Copyright{David Eppstein, Daniel Frishberg, and Elham Havvaei} %TODO mandatory, please use full first names. LIPIcs license is "CC-BY";  http://creativecommons.org/licenses/by/3.0/

\ccsdesc[500]{Theory of computation~Design and analysis of algorithms}
\keywords{directed acyclic graph,  activity-on-edge graph, critical path, project planning, milestone minimization, graph visualization} %TODO mandatory; please add comma-separated list of keywords

\category{} %optional, e.g. invited paper

\relatedversion{} %optional, e.g. full version hosted on arXiv, HAL, or other respository/website
\supplement{}%optional, e.g. related research data, source code, ... hosted on a repository like zenodo, figshare, GitHub, ...

\funding{This work was supported in part by NSF grants  CCF-1618301 and CCF-1616248.}%optional, to capture a funding statement, which applies to all authors. Please enter author specific funding statements as fifth argument of the \author macro.

\acknowledgements{}%optional

\nolinenumbers %uncomment to disable line numbering

\hideLIPIcs  %uncomment to remove references to LIPIcs series (logo, DOI, ...), e.g. when preparing a pre-final version to be uploaded to arXiv or another public repository

\EventEditors{}
\EventLongTitle{} 
\EventShortTitle{}
\EventAcronym{}
\EventYear{}
\EventDate{}
\EventLocation{}
\EventLogo{}
\SeriesVolume{}
\ArticleNo{}
\usepackage{hyperref,microtype,cite}
\usepackage{graphicx}
\usepackage{xspace}

\usepackage{amssymb}
\usepackage{amsmath}

\usepackage[ruled, vlined,linesnumbered,noresetcount]{algorithm2e}
\graphicspath{ {} }

\newcommand{\stv}{\operatorname{St}}
\newcommand{\en}{\operatorname{End}}
\newcommand{\opt}{\textit{Opt}\xspace}
\newcommand{\alg}{\textit{$\mathcal{A}$}\xspace}

\newcommand{\rsag}[1]{\rightsquigarrow_{#1}}
\begin{document} 
\maketitle           
\begin{abstract}
We formalize the simplification of \emph{activity-on-edge} graphs used for visualizing project schedules,
where the vertices of the graphs represent project milestones, and the edges represent either
tasks of the project or timing constraints between milestones.
In this framework, a timeline of the project can be constructed as a leveled drawing of the graph,
where the levels of the vertices represent the time at which each milestone is scheduled to happen.
We focus on the following problem: given an activity-on-edge graph representing a project,
find an equivalent activity-on-edge graph\textemdash one with the same critical paths\textemdash
that has the minimum possible number of milestone vertices among all equivalent activity-on-edge graphs.
We provide a polynomial-time algorithm for solving this graph minimization problem.
\end{abstract}

\section{Introduction}
\label{sec:intro}
    
\begin{figure}[t]
 	\center
    \includegraphics[scale=0.8]{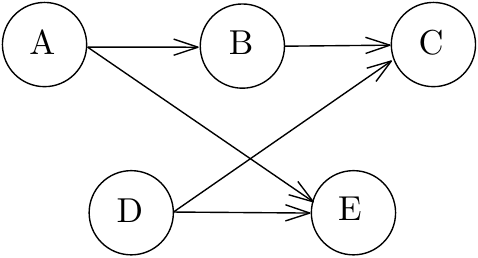} 
\smallskip
 
     \includegraphics[scale=0.6]{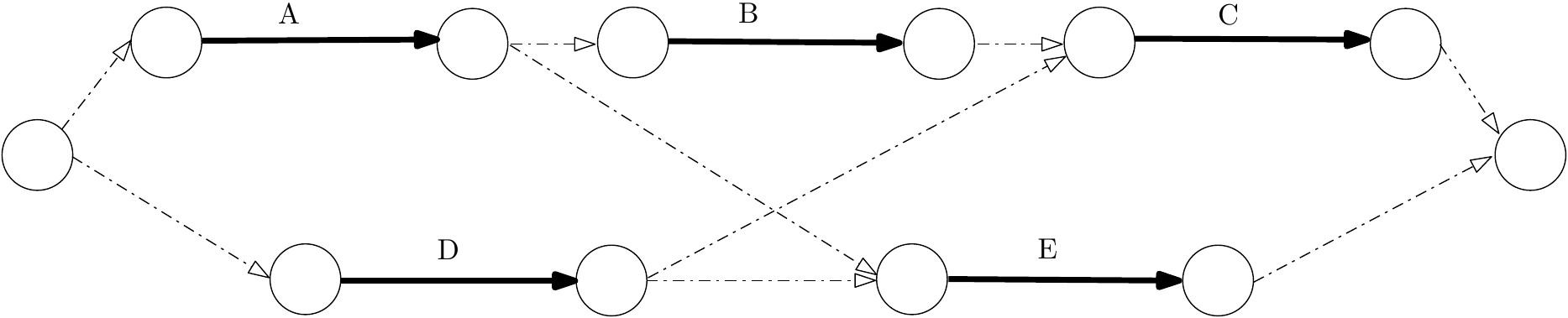}
  
  \caption{An activity-on-node graph, above, and its naively expanded activity-on-edge graph, below, with solid arrows as task edges and dashed-dotted arrows as unlabeled edges. }
  \label{fig:aonaoe}
\end{figure}

The \emph{critical path method} is used in project modeling to describe the tasks of a project, along with the dependencies among the tasks; it was originally developed as PERT by the United States Navy in the 1950s~\cite{malcolm1959application}.  A dependency graph is used to identify bottlenecks, and in particular to find the longest path among a sequence of tasks, where each task has a required length of time to complete (this is known as the \emph{critical path}).

In this paper we are interested in the problem of visualizing an abstract timeline of the possible critical paths of a given project, represented abstractly as a partially ordered set of tasks, at a point in the project planning at which we do not yet know the time lengths of each task.
The most common method of visualizing partially ordered sets, as an \emph{activity-on-node graph} (a transitively reduced directed acyclic graph with a vertex for each task) is unsuitable for this aim, because it represents each task as a point instead of an object that can extend over a span of time in a timeline. To resolve this issue, we choose to represent each task as an edge in a directed acyclic graph. In this framework, the endpoints of the task edges have a natural interpretation, as the \emph{milestones} of the project to be scheduled. Additional \emph{unlabeled edges} do not represent tasks to be performed within the project, but constrain certain pairs of milestones to occur in a certain chronological order. The resulting \emph{activity-on-edge graph} can then be drawn in standard upward graph drawing style~\cite{bt1992area,battista1998graph,gt2001computational,gt1995upward,ahr2010improving}.
Alternatively, once the lengths of the tasks are known and the project has been scheduled, this graph can be drawn in leveled style~\cite{junger1999level,hkl2004characterization}, where the level of each milestone vertex represents the time at which it is scheduled.

It is straightforward to expand an activity-on-node graph into an activity-on-edge graph
by expanding each task vertex of the activity-on-node graph into a pair of milestone vertices connected by a task edge, with the starting milestone of each task retaining all of the incoming unlabeled edges of the activity-on-node graph and the ending milestone retaining all of the outgoing edges. It is convenient to add two more milestones at the start and end of the project, connected respectively to all milestones with no incoming edges and from all milestones with no outgoing edges. The size of the resulting activity-on-edge graph is linear in the size of the activity-on-node graph. An example of such a transformation is depicted in \autoref{fig:aonaoe}.

However, the  graphs that result from this naive expansion are not minimal. Often, one can merge some pairs of milestones (for instance the ending milestone of one task and the starting milestone of another task) to produce a simpler activity-on-edge graph (such as the one for the same schedule in \autoref{fig:aoesimplified}). Despite having fewer milestones, this simpler graph can be equivalent to the original, in the sense that its critical paths (maximal sequences of tasks that belong to a single path in the graph) are the same. By being simpler, this merged graph should aid in the visualization of project schedules.
In this paper we formulate and provide a polynomial time algorithm for the problem of optimal simplification of activity-on-edge graphs.

\subsection{New Results}
\begin{figure}
	\center
    \includegraphics[scale=0.8]{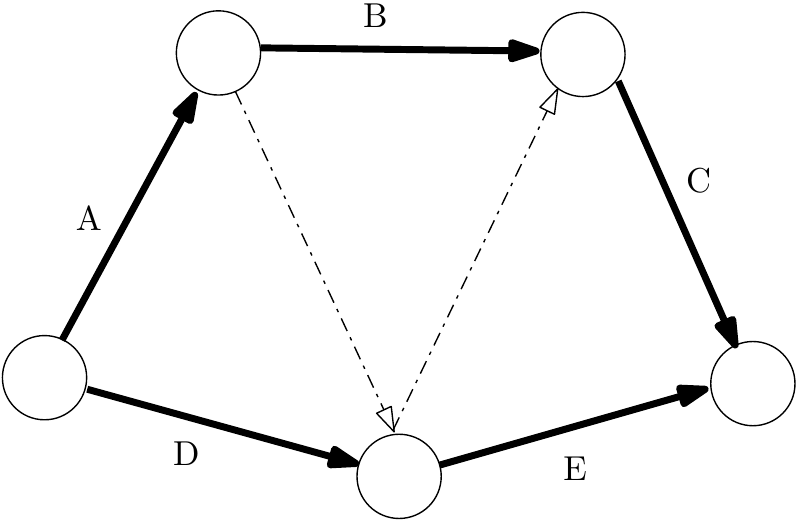}
  \caption{ A simplification of the graph from \autoref{fig:aonaoe}.}
  \label{fig:aoesimplified}
\end{figure}
We describe a polynomial-time algorithm which, given an activity-on-edge graph (i.e., a directed acyclic graph with a subset of its edges labeled as tasks), produces a directed acyclic graph
that preserves the critical paths of the graph and has the minimum possible number of
vertices among all critical-path-preserving graphs for the given input.
Our algorithm is agnostic about the weights of the tasks. In more general terms, the resulting graph has the following properties:
\begin{itemize}
\item The task edges in the given graph correspond one-to-one with the task edges in the new graph.
\iffalse \item For any positive weight function on the task edges (with zero weights on the unlabeled edges), the new graph has the same longest path as the original graph. \fi
\item The new graph has the same dependency (reachability) relation among task edges as the original graph.
\item The new graph has the same potential critical paths as the original graph.
\item The number of vertices of the graph is minimized among all graphs with the first three properties.
\end{itemize}
\iffalse Since the shortest possible completion time of a project is the longest path in an activity-on-edge graph, the resulting graph preserves a crucial property of the original graph while simplifying its structure. \fi

Our algorithm repeatedly applies a set of local reduction rules, each of which either merges a pair of adjacent vertices or removes an unlabeled edge, in arbitrary order. When no rule can be applied, the algorithm outputs the resulting graph.

We devote the rest of this section to related work and then describe the preliminaries in \autoref{sec:prelim}. We then present the algorithm in \autoref{sec:alg} and show in \autoref{sec:correctness} that its output preserves the potential critical paths of the input, and in \autoref{sec:optimality} that it has the minimum possible number of vertices. We also show that the output is independent of the order in which the rules are applied. We discuss the running time in \autoref{sec:analysis} and conclude with \autoref{sec:conclusion}.

\subsection{Related work}
Constructing clear and aesthetically pleasing drawings of directed acyclic graphs is an old and well-established task in graph drawing, with many publications~\cite{sugiyama1981methods,battista1998graph,bm2001layered,hn2014hierarchical}.
The work in this line that is most closely relevant for our work involves
upward drawings of unweighted directed acyclic graphs~\cite{bt1992area,gt2001computational,gt1995upward,ahr2010improving}
or leveled drawings of directed acyclic graphs that have been given a level assignment~\cite{junger1999level,hkl2004characterization}
(an assignment of a $y$-coordinate to each vertex, for instance representing its height on a timeline).

Although multiple prior publications use activity-on-edge graphs~\cite{kulkarni1984compact,cg1986parallel,arh1994nearcritical,hd2011project}
and even consider graph drawing methods specialized for these graphs~\cite{xu2010automatic},
we have been unable to locate prior work on their simplification.
This problem is related to a standard computational problem, the construction of the transitive reduction of a directed acyclic graph or equivalently the covering graph of a partially ordered set~\cite{agu1972transitive}. We note in addition our prior work on augmenting partially ordered sets with additional elements (preserving the partial order on the given elements) in order to draw the augmented partial order as an upward planar graph with a minimum number of added vertices~\cite{es2013confluent}.

The PERT method may additionally involve the notion of ``float'', in which a given task may be delayed some amount of time (depending on the task) without any effect on the overall time of the project~\cite{arditi2006selecting,householder1990owns}. We do not consider constraints of this form in the present work, although the unlabeled edges of our output can in some sense be seen as serving a similar purpose.

\section{Preliminaries}
\label{sec:prelim}

We first define an activity-on-edge graph. The graph can be a multigraph to allow tasks that can be completed in parallel to share both a start and end milestone when possible.

\begin{definition}
\label{def:CPG}
Define an \emph{activity-on-edge graph (AOE)} as a directed acyclic multigraph $G = (V, E)$, where a subset of the edges of $E$, denoted $\mathcal{T}$, are labeled as \emph{task edges}. The labels, denoting tasks, are distinct, and we identify each edge in $\mathcal{T}$ with its label.
\end{definition}

\begin{definition}
\label{def:stend}
Given an AOE $G$ with tasks $\mathcal{T}$, for all $T \in \mathcal{T}$, let $\stv_G(T)$ be the start vertex of $T$, and let $\en_G(T)$ be the end vertex of $T$.
\end{definition}

When the considered graph is clear from context, we omit the subscript $G$ and write $\stv(T)$ and $\en(T)$.
It may be that $\stv(T) = \stv(T')$, or $\en(T) = \en(T')$, or $\en(T) = \stv(T')$ with $T \neq T'$.

To define potential critical paths formally, we introduce the following notation.
\begin{definition}
\label{def:haspath}
Given an AOE $G$ with tasks $\mathcal{T}$, for all $T$, $T'$ with $T \neq T'$, say that $T$ \emph{has a path to} $T'$ in $G$ if there exists a path from $\en(T)$ to $\stv(T')$, or if $\en(T) = \stv(T')$, and write $T \rsag{G} T'$.
\end{definition}

\begin{definition}
\label{def:critpath}
Given an AOE $G$ with tasks $\mathcal{T}$, define a \emph{potential critical path} as a sequence of tasks $P =(T_1, \dots, T_k)$, where for all $i = 1, \dots, k - 1$, $T_i \rsag{G} T_{i+1}$, and where $P$ is not a subsequence of any other sequence with this property.
\end{definition}

Our algorithm will apply a set of transformation rules to an input AOE of a \emph{canonical} form.
\begin{definition}
\label{def:inputCPG}
A \emph{canonical AOE} is an AOE which is naively expanded from an activity-on-node graph. 
\end{definition}

Every AOE $G$ can be transformed into a canonical AOE with the same reachability relation on its tasks. First, we start by computing the reachability relation of the tasks. The transitive closure of the resulting reachability matrix gives an activity-on-node graph (which is quadratic, in the worst case, in the size of the original AOE). Then, this activity-on-node graph can be converted to a canonical AOE as described in Section~\ref{sec:intro}. 

\begin{definition}
\label{def:equivaoe}
Two AOE graphs $G$ and $H$ are \emph{equivalent}, i.e. $G \equiv H$, if $G$ and $H$ have the same set of tasks\textemdash i.e., there is a label-preserving bijection between the task edges of $G$ and those of $H$\textemdash and, with respect to this bijection, $G$ and $H$ have the same set of potential critical paths.
\end{definition}

\begin{definition}
\label{def:optaoe}
An AOE $G$ is \emph{optimal} if $G$ minimizes the number of vertices for its equivalence class: i.e., if for every AOE $H \equiv G$, $|V(H)| \geq |V(G)|$.
\end{definition}

\iffalse
\begin{definition}
\label{def:optCPG}
Given an AOE $G$, let an \emph{optimal AOE} (\opt) be any AOE with the same set of potential critical paths as $G$, with the minimum possible number of vertices. 
\end{definition}
\fi

We now formally define our problem.
\\
\textbf{Problem 1.} \textit{Given a canonical AOE $G$, find some optimal AOE $H$ with $H \equiv G$.}

\section{Simplification Rules}
\label{sec:alg}

\begin{figure}
\center
    \includegraphics[scale=0.75]{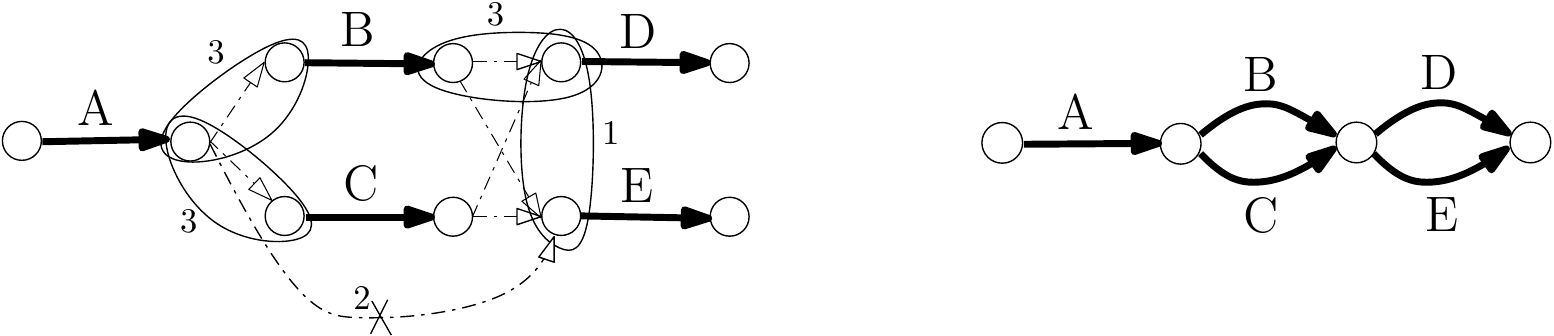}
  \caption{On the left, an AOE in which each of rules~\ref{rule:SameNeigh}-\ref{rule:AllHasPath} can be applied, and on the right, the corresponding graph output by the algorithm.}
  \label{fig:ruleexample}
\end{figure}
Our algorithm takes a canonical AOE and greedily applies a set of rules until no more rules can be applied. Given an AOE $G = (V, E)$ and given two distinct vertices $u, v \in V$, the simplification rules used by our algorithm are:

\begin{enumerate}
    \item if $u$ and $v$ have no outgoing task edges and have precisely the same outgoing neighbors, merge $u$ and $v$. Symmetrically, if $u$ and $v$ have no incoming task edges and have precisely the same incoming neighbors, merge $u$ and $v$.
    \label{rule:SameNeigh}
   % \item If $u$ has only one outgoing edge, and it is an unlabeled edge going to vertex $v$, merge $u$ and $v$. Symmetrically, merge $u$ and $v$ if $v$ has only one incoming edge, it is unlabeled, and it comes from $u$.
 %   \label{rule:OneUnl}
    \item If $u$ has an unlabeled edge to $v$, and $u$ has another path to $v$, remove the edge $(u, v)$.
    \label{rule:HasPath}
    \item If $u$ has an unlabeled edge to $v$ and the following conditions are satisfied, merge $u$ and $v$:
    	\begin{itemize}
    		\item rule \ref{rule:HasPath} is not applicable to the edge $(u,v)$.
    		\item if $u$ has an outgoing task, then $v$ has no incoming edge other than $(u,v)$.
    		\item if $v$ has an incoming task, then $u$ has no outgoing edge other than $(u,v)$.
    		\item every incoming neighbor of $v$ has a path to every outgoing neighbor of $u$.
    	\end{itemize}
        \label{rule:AllHasPath}
        
\end{enumerate}

\autoref{fig:ruleexample} depicts an AOE graph and the graph output by the algorithm after applying all possible rules. 

It will be convenient for the proofs in Section~\ref{sec:optimality} to give a name to the output of the algorithm:
\begin{definition}
\label{def:outputaoe}
An \emph{output AOE}, denoted \alg, is any AOE obtained from a canonical AOE $G$ by a sequence of applications of rules~\ref{rule:SameNeigh}, \ref{rule:HasPath}, and \ref{rule:AllHasPath}, to which none of these rules can still be applied.
\end{definition}

\section{Correctness}
\label{sec:correctness}

In this section we prove the correctness of our algorithm (its output graph is equivalent to its input graph).

We begin with preserving potential critical paths. We show that the rules never change the existence or nonexistence of a path from one task to another, and that this implies preservation of potential critical paths.

\begin{lemma}
\label{lem:cpimpliespo}
Given two AOEs $G$ and $H$ with the same set of tasks $\mathcal{T}$, $G$ and $H$ have the same reachability relation $\rsag{}$ on the tasks if and only if $G \equiv H$.
\end{lemma}
\begin{proof}
Trivially, we have $T\rsag{G} T'$ (or $T\rsag{H} T'$) if and only if $T$ is earlier than $T'$ in some potential critical path of $G$ (or $H$). Therefore, preservation of critical paths is equivalent to preservation of the reachability relation.  
\end{proof}

\begin{lemma}
\label{lem:doesnotremovepath}
\iffalse Given tasks $T$ and $T'$, we have $T \rsag{} T'$ at a given iteration of the algorithm if and only if $T \rsag{} T'$ at the next iteration. \fi
The output of the algorithm is equivalent to its input.
\end{lemma}
\begin{proof}
We show the invariant that given tasks $T$ and $T'$, $T \rsag{} T'$ at a given iteration of the algorithm if and only if $T \rsag{} T'$ at the next iteration. From this it follows that the output of the algorithm has the same reachability relation on its tasks as the input, and then the lemma follows from Lemma~\ref{lem:cpimpliespo}.

The invariant is true because merging a pair of vertices (rules \ref{rule:SameNeigh} and \ref{rule:AllHasPath}) never disconnects a path, and no edge is ever removed (by rule \ref{rule:HasPath}) between two vertices unless another path exists between the two vertices. In particular, the end vertex of $T$ still has a path to the start vertex of $T'$ after the application of any of the rules.

For the other direction, removing an edge never introduces a new path. Furthermore, if vertices $u$ and $v$ are merged by applying rule~\ref{rule:SameNeigh}, and if some vertex $w$ has a path to some vertex $z$ through the newly merged $uv$, then the condition of rule~\ref{rule:SameNeigh} ensures that $w$ has a path, through $u$ or $v$, to $z$ before the merge. Similarly, suppose $u$ and $v$ are merged by applying rule~\ref{rule:AllHasPath}. Then if $w$ has a path to $z$ through $uv$, then (abusing notation) either $w \rsag{} u$ and $v \rsag{} z$ before the merge, so $w \rsag{} z$ (via the edge $(u, v)$), or for some incoming neighbor $x$ of $v$ and outgoing neighbor $y$ of $u$, $w \rsag{} x$ and $y \rsag{} z$. In this case, by the conditions of the rule, $w \rsag{} z$ before the merge.
\end{proof}

\begin{restatable}{lemma}{lemcycle}
\label{lem:nocycle-intermediate}
Any intermediate graph that results from applying rules of the algorithm to an input canonical AOE graph, is acyclic.
\end{restatable}
\begin{proof}
 Given~\autoref{def:CPG} and~\autoref{def:inputCPG}, the canonical AOE input $G$ is acyclic. Now we show none of the rules can create a cycle after being applied to an intermediate acyclic graph $G'$. This is obvious for rule~\ref{rule:HasPath} as it removes edges.  Suppose for a contradiction that merging vertices $u$ and $v$ creates a cycle. The cycle must involve the new vertex resulting from the merge. For rule~\ref{rule:SameNeigh}, this implies the existence of a cycle in $G'$ either from $u$ or $v$ to itself which is a contradiction. For rule~\ref{rule:AllHasPath}, it implies the existence of a cycle in $G'$ including the unlabeled edge $(u,v)$ or a cycle including an incoming neighbor of $v$ and an outgoing neighbor of $u$, which is a contradiction.  
\end{proof}
\begin{corollary}
Any graph $\mathcal{A}$ output by the algorithm is acyclic. 
\label{cor:nocycle}
\end{corollary}
\section{Optimality}
\label{sec:optimality}

In this section we prove the optimality of our algorithm: it uses as few vertices as possible.
Let \alg be any output AOE. Let \opt be any optimal AOE such that $\alg \equiv \opt$. Our proof relies on an injective mapping from the vertices of \alg to the vertices of \opt. \iffalse To define this mapping, we rely on two key facts about the graphs. First, it is immediate from \autoref{def:critpath} that every task is part of some critical path. Therefore, every task must be in \opt, and (as our algorithm never removes a task)
it is also present in \alg. Second, as we will show our algorithm produces a graph in which every vertex has an incident task edge. The mapping is then defined as follows: given a vertex in \alg, check whether it has an incoming task or an outgoing task (or both). If the vertex has an outgoing task, map it to the start vertex of the same task in \opt. If the vertex has no outgoing task, map it to the end vertex of its incoming task in \opt. The two facts above make this mapping well-defined.
To show that this mapping is one-to-one, we show two facts: first, if in \alg two tasks have distinct start (end) vertices, then the same tasks have distinct start (end) vertices in \opt. Second, if in \alg a task $T$ has its end vertex distinct from the start vertex of task $T'$, then the same fact holds for $T$ and $T'$ in \opt.

\fi The existence of this mapping shows that \alg has at most as many vertices as \opt, and therefore has the optimal number of vertices.
Once we have identified the vertices of \alg with the vertices of \opt in this way, we show that, for a given input, any two graphs output by the algorithm (but not necessarily \opt) must have the same unlabeled edges. Since the task edges are determined, and since the injective mapping to \opt determines the vertices, determining the unlabeled edges implies the order-independence of our algorithm's choice of simplification rules.

Before defining the mapping between \alg and \opt, we establish some facts about the structure of \alg.

\begin{lemma}
\label{lem:taskdensity}
For every unlabeled edge $(u, v)$ in any output AOE \alg, there exist tasks $T$ and $T'$ such that $u = \en(T)$ and $v = \stv(T')$.
\end{lemma}
\begin{proof}
By Definition~\ref{def:outputaoe}, \alg is produced by the algorithm from some canonical AOE $G$. This property holds for $G$ by \autoref{def:inputCPG}. As every rule of the algorithm either removes an unlabeled edge or merges two vertices, and never creates a new edge or vertex, the proof is complete. 
\end{proof}

\begin{corollary}
\label{cor:taskforvertex}
Every vertex in an output AOE \alg has an incident task edge.
\end{corollary}

\iffalse
\begin{lemma}
\label{lem:tasksinboth}
Given a canonical AOE $G$ with tasks $\mathcal{T}$, every task in $\mathcal{T}$ has a corresponding task edge both in \alg and in \opt.
\end{lemma}
\begin{proof}
The algorithm never removes task edges, so every task in $\mathcal{T}$ is also in \alg. Furthermore, every task $T$ is part of some potential critical path in $G$ (\autoref{def:critpath} implies every task is in some potential critical path), so \opt must have a task edge labeled $T$. 
\end{proof}
\fi

We can now define a mapping from the vertices of \alg to those of \opt:

\begin{definition}
\label{def:mapping}
Given an output AOE \alg with task set $\mathcal{T}$, and given an optimal AOE \opt with $\alg \equiv \opt$, let $M: V(\alg) \rightarrow V(\opt)$ be the following mapping: for every $v~\in~V(\alg)$:\begin{itemize}
    \item Let $M(v) = \stv_{\opt}(T)$, for some $T \in \mathcal{T}$ for which $v = \stv_{\alg}(T)$, if such a task exists.
    \item Let $M(v) = \en_{\opt}(T)$, where $v = \en_{\alg}(T)$, otherwise.
\end{itemize}
\end{definition}
As shown in \autoref{cor:taskforvertex}, every vertex in \alg has an incident task edge, and by Definition~\ref{def:equivaoe}, \alg and \opt have the same set of tasks. Therefore, this mapping is well-defined (up to its arbitrary choices of which task to use for each $v$). To prove that $M$ is injective, we will use the fact that since $\alg \equiv \opt$, \alg and \opt have the same reachability relation (by \autoref{lem:cpimpliespo} and \autoref{lem:doesnotremovepath}).

\iffalse
Before showing that $M$ is injective, we need one more fact: that there are no cycles in \alg.
\fi

The heart of the proof that $M$ is injective lies in showing that if two tasks do not share a vertex in \alg, the corresponding tasks also do not share the corresponding vertices in \opt. From this it follows that $M$ cannot map distinct vertices in \alg to the same vertex in \opt.
\begin{lemma}
\label{lem:startstartvertices}

Given an output AOE \alg, and an optimal AOE $\opt \equiv \alg$, with task set $\mathcal{T}$, let $T$ and $T'$ be two distinct tasks in $\mathcal{T}$. If $\stv_\alg(T) \neq \stv_\alg(T')$, then $\stv_\opt(T) \neq \stv_\opt(T')$. If  $\en_\alg(T) \neq \en_\alg(T')$, then $\en_\opt(T) \neq \en_\opt(T')$.
\end{lemma}
\begin{proof}

Suppose for a contradiction that $\stv_\alg(T)$ $\neq$ $\stv_\alg(T')$, but $\stv_\opt(T) = \stv_\opt(T')$ (the other case is symmetrical). Let $u = \stv_\alg(T)$ and $v = \stv_\alg(T')$. Consider the following (exhaustive) cases for $u$ and $v$:
\begin{enumerate}
    \item $u$ and $v$ have no incoming edges
    \label{case:NoEdge}
    \item $u$ or $v$ has an incoming unlabeled edge, but neither $u$ nor $v$ has an incoming task edge
    \label{case:JustUnl}
    \item $u$ or $v$ has an incoming task edge $A$
    \label{case:haveTask}
\end{enumerate} 

In case~\ref{case:NoEdge}, applying rule~\ref{rule:SameNeigh} results in merging $u$ and $v$. However, since \alg is the output of the algorithm, no rule can be applied to \alg. This is a contradiction.

In case~\ref{case:JustUnl}, $u$ and $v$ cannot have the same incoming neighbors or else rule~\ref{rule:SameNeigh} would apply. We may assume without loss of generality that there exist a vertex $w$ and an unlabeled edge $(w,u)$, such that the edge $(w, v)$ does not exist. By \autoref{lem:taskdensity}, there exists a task $A$ where $w = \en_\alg(A)$. Since $A \rsag{\alg} T$ and $\alg \equiv \opt$ (by Lemma~\ref{lem:doesnotremovepath}), then by Lemma~\ref{lem:cpimpliespo}, $A \rsag{\opt} T$, so $A \rsag{\opt} T'$, since $\stv_\opt(T) = \stv_\opt(T')$. Again by Lemma~\ref{lem:cpimpliespo}, $A \rsag{\alg} T'$, so there is a path $P$ from $w$ to $v$. If $|P| = 1$, then this contradicts that $(w, v)$ does not exist.
Suppose $|P| > 1$. Then we show there exist some vertex $w' \neq w$ and an unlabeled edge $(w',v)$. The following cases are exhaustive:
 \begin{enumerate}[(a)]
 	\item $P$ contains a path from $u$ to $v$. As such a path to $v$ exists and $v$ has no incoming task edge, there exist a vertex $w'$ and an unlabeled edge $(w',v)$ ($w' \neq u$), not belonging to $P$ unless rule~\ref{rule:AllHasPath} can be applied to vertex $v$ and its incoming neighbor in path $P$ .  
 	\label{case:w'_outofP}
 	\item $P$ does not contain a path from $u$ to $v$. As $|P| > 1$, an unlabeled edge $(w',v)$ belonging to path $P$ exists. 
 	\label{case:w'_belongP}
 \end{enumerate}
 Given the existence of $(w',v)$, by \autoref{lem:taskdensity}, there exists a task $B$ where $w' = \en_\alg(B)$. $B \rsag{\alg} T'$, so by reasoning similar to the above, $B \rsag{\alg} T$. Then, one can apply rule~\ref{rule:HasPath} and either remove edge $(w',v)$ in case~\ref{case:w'_outofP} or $(w,u)$ in case~\ref{case:w'_belongP} (Figure~\ref{fig:startstartcase2}); this contradicts the definition of \alg. 

In case~\ref{case:haveTask}, we can assume without loss of generality that $u$ has an incoming task $A$; consequently, $u = \en_\alg(A)$. Then, by Lemma~\ref{lem:cpimpliespo}, we have  $A \rsag{\opt} T'$ and thus $A \rsag{\alg} T'$ via a path $P$. Consider the following cases for $P$:
    \begin{enumerate}[(a)]
        \item $P$ contains a task edge $B$ 
        \item $P$ is a sequence of unlabeled edges
    \end{enumerate}
    
In case a, by Lemma~\ref{lem:cpimpliespo}  $B \rsag{\opt} T'$, and thus $B \rsag{\opt} T$, and therefore $B \rsag{\alg} T$. This creates a cycle between $u$ and  $\en_\alg(B)$, contradicting \autoref{cor:nocycle}.

In case b, illustrated in ~\autoref{fig:startstartcase3}, since rule~\ref{rule:AllHasPath} cannot be applied (if it could, this would contradict the definition of \alg), there exist a vertex $x$ not on the path from $u$ to $v$, and an edge $(x,v)$ (a task edge or an unlabeled edge). Therefore, there exists a task $B$ where either $v = \en_\alg(B)$ or by \autoref{lem:taskdensity}, $x = \en_\alg(B)$. Considering \opt and applying \autoref{lem:cpimpliespo}, $B \rsag{\opt} T$ so $B \rsag{\alg} T$. This path either creates a cycle in \alg or allows for removing edge $(x,v)$ by rule~\ref{rule:HasPath}, which is a contradiction.

Thus if $\stv_\alg(T) \neq \stv_\alg(T')$, then $\stv_\opt(T) \neq \stv_\opt(T')$. 
\begin{figure}[t]
\center
  \includegraphics[scale=0.8]{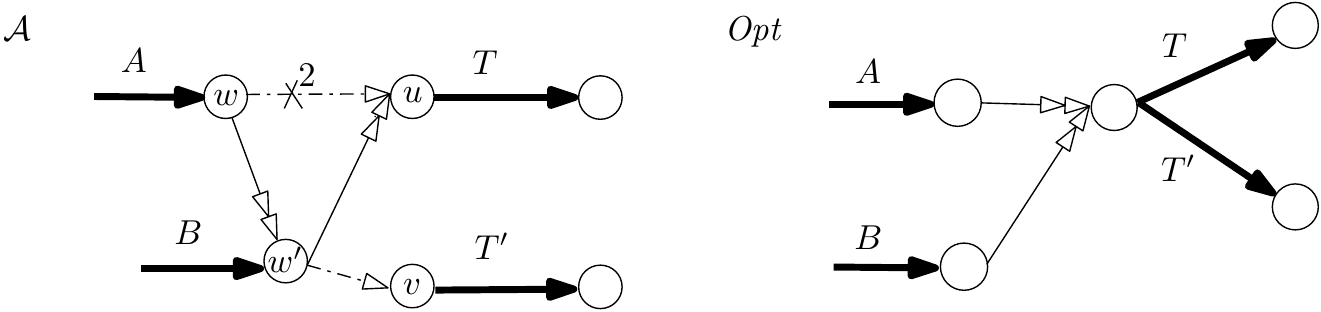}
  \caption{\autoref{lem:startstartvertices}, case~\ref{case:JustUnl}, subcase~\ref{case:w'_belongP} (double arrows indicate paths).}
  \label{fig:startstartcase2}

  \includegraphics[scale=0.8]{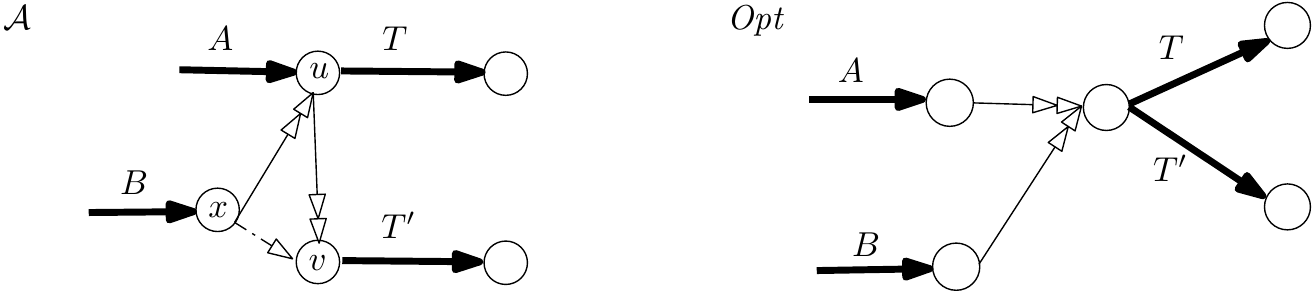}
  \caption{\autoref{lem:startstartvertices}, case~\ref{case:haveTask}, subcase b.}
  \label{fig:startstartcase3}
\end{figure}

\end{proof}

\begin{restatable}{lemma}{endstartvertices}

\label{lem:endstartvertices}
Given an output AOE \alg, and an optimal AOE $\opt \equiv \alg$, with task set $\mathcal{T}$, let $T$ and $T'$ be two distinct tasks in $\mathcal{T}$. If $\en_\alg(T) \neq \stv_\alg(T')$, then $\en_\opt(T) \neq \stv_\opt(T')$.
\end{restatable}
\begin{proof}
The proof, which is in Appendix~\ref{app:endstartvertices}, uses essentially the same approach as the proof of Lemma~\ref{lem:startstartvertices}: supposing that the two vertices are the same, then using the fact that \alg and \opt have the same reachability relation on their tasks, and the definition of \alg as having no more rules to apply, to derive a contradiction. 

\end{proof}

%clearly no cycle in an optimal AOE \opt contains a task, since \autoref{lem:cpimpliespo} implies a well-defined reachability relation among the task edges in \opt. Furthermore, if \opt contains an unlabeled cycle, then the vertices in the cycle can simply be merged to produce a graph with the same reachability relation and fewer vertices. 

There is one remaining technicality: we have defined an optimal AOE as being acyclic; the question arises whether one could reduce the number of vertices by allowing (unlabeled) cycles. However, this is not the case; it is easy to see that any unlabeled cycle can be merged into one vertex, reducing the number of vertices, without changing the reachability relation on the tasks. 

We are ready to prove our main results.
\begin{theorem}
\label{thm:endisopt}
Given a canonical AOE $G$, the algorithm produces an optimal AOE $\opt \equiv G$.
\end{theorem}
\begin{proof}
Let \alg be the output AOE produced by the algorithm on $G$. Given any optimal AOE \opt and \alg, the mapping $M$ in \autoref{def:mapping} is injective: suppose for a contradiction that $u$ and $v$ are distinct vertices in \alg, and $w = M(u) = M(v)$. Then by the definition of $M$, either $u$, $v$, and $w$ have the same incoming task, or $u$, $v$, and $w$ have the same outgoing task, or there exist tasks $T$ and $T'$ such that (without loss of generality) $u = \en_\alg(T), v = \stv_\alg(T')$, and $\en_\opt(T) = w = \stv_\opt(T')$. By Lemmas~\ref{lem:startstartvertices} and \ref{lem:endstartvertices}, all three of these cases imply that $u = v$.

Therefore, $|V(Opt)| = |V(\alg)|$. Furthermore, $\alg \equiv G$, by Lemma~\ref{lem:doesnotremovepath}. The theorem follows. 
\end{proof}

\begin{theorem}
\label{thm:endisunique}
Given an input, the algorithm produces the same output regardless of the order in which the rules are applied.
\end{theorem}
\begin{proof}
As stated earlier, all task edges of an input canonical AOE $G$ are present in any output of the algorithm and the mapping determines the vertices. Therefore, it suffices to show that any two graphs output by the algorithm have the same set of unlabeled edges. Suppose for a contradiction that $\alg_1$ and $\alg_2$ are two distinct outputs of the algorithm, resulting from applying different sequences of rules. By \autoref{thm:endisopt}, the algorithm always produces an optimal AOE. Therefore, $|V_{\alg_1}| = |V_{\alg_2}| = |V_\opt|$. Since $\alg_1 \neq \alg_2$, there is an unlabeled edge $(u,v)$ in $\alg_1$ (without loss of generality) that is not in $\alg_2$. By \autoref{lem:taskdensity}, there exist task edges $T$ and $T'$ such that $u = \en_{\alg_1}(T)$ and $ v = \stv_{\alg_1}(T')$. We have  $T \rsag{\alg_1} T'$. Since by \autoref{lem:cpimpliespo} and \autoref{lem:doesnotremovepath}, $\alg_1$ and $\alg_2$ both preserve the reachability relation of the tasks of $G$, we have $T \rsag{\alg_2} T'$. Consider the cases for path $P$ from $T$ to $T'$ in $\alg_2$:

\begin{enumerate}
    \item There exists a task $A$ in $P$ other than $T$ and $T'$.
    \label{uniq:OnlyT}
    \item Path $P$ is a sequence of unlabeled edges.
    \label{uniq:OnlyU}
\end{enumerate}
 
In case~\ref{uniq:OnlyT}, we have $T \rsag{\alg_2} A \rsag{\alg_2} T'$ and therefore, $T \rsag{\alg_1} A \rsag{\alg_1} T'$. Then by rule~\ref{rule:HasPath}, one can remove the edge $(u,v)$, which contradicts the definition of $\alg_1$.

In case~\ref{uniq:OnlyU}, the length of $P$ is at least two, and $P$ contains a vertex $w$. By \autoref{lem:taskdensity}, there exist tasks $A$ and $B$ where $w = \en_{\alg_2}(A) = \stv_{\alg_2}(B)$. Now, since $\alg_1 \equiv \alg_2$, both graphs are optimal, and both graphs are outputs of the algorithm, Lemma~\ref{lem:endstartvertices} implies that $\en_{\alg_1}(A) = \stv_{\alg_1}(B)$. Call this vertex $x$. Then there exists a path from $u$ to $v$, through $x$, by Lemma~\ref{lem:cpimpliespo}, and one can remove the edge $(u, v)$ by rule~\ref{rule:HasPath}. This contradicts the definition of $\alg_1$.
\end{proof}

\section{Analysis}
\label{sec:analysis}
Let $n$ be the number of vertices in a canonical AOE (which is linear in the number of tasks), and $m$ the number of unlabeled edges. There are at most $O(n + m)$ iterations in the algorithm, because each iteration either merges two vertices or removes an edge, by applying one of the three rules. This requires finding an edge to remove ($O(m)$ potential edges) or two vertices to merge ($O(n^2)$ potential pairs), then performing the merge or the removal. Intuitively, our algorithm runs in polynomial time as it takes polynomial time to find and apply a rule. 

We provide a faster implementation of our algorithm than the naive approach. The algorithm transforms a canonical AOE graph $G$ into an optimal AOE graph by applying rules 1, 2 or 3. For simplicity, we label the vertices $1,\dots, n$. At each iteration, compute a reachability matrix $M$ for the current graph. $M[u][v]$ indicates whether there exist zero, one, or more than one paths from $u$ to $v$. In order to compute $M$, for all $u$ and $v$ initialize $M[u][v] = 1$ if the edge $(u,v)$ exists. Then sort the vertices in topological order (such an ordering exists according to~\autoref{lem:nocycle-intermediate}). For each vertex $v$ in this order, and for each vertex $u$, set  $M[u][v]$ to $\min(2, \sum_{w\in W}(M[u][w]))$, where $W$ is the set of all vertices $w$ such that either $w = v$ or there exists an edge $(w,v)$. This procedure takes $O(nm)$ time. Algorithm 1 provides a summary.

\begin{algorithm}[t]
 \KwData{Canonical AOE $G$}
 \KwResult{Optimal AOE \opt}
 \While{true}{
  Initialize and compute the reachability matrix $M$\;
  Remove, by rule 2, all unlabeled edges $(u,v)$ where $M[u][v] = 2$\;
 \uIf{rule 1 applies}{
     apply rule 1\;
  }
  \uElseIf{rule 3 applies}{
     apply rule 3\;
  }
  \Else{
    return the graph\;
  }
 }
 \caption{the proposed transformation algorithm.}
\end{algorithm}

Given the reachability matrix, an unlabeled edge $(u,v)$ is removed by rule 2 in $O(1)$ time, if $M[u][v] \geq 2$. Therefore, checking rule 2 for all edges takes $O(m)$ time.

Without loss of generality, for rule 1, we only consider merges of pairs of vertices with the same outgoing neighbors. This requires, for each vertex $u$ with no outgoing task edge, a sorted list of outgoing neighbors (S$[u]$). To obtain such lists for all vertices, list unlabeled edges as pairs of vertices and sort all the pairs with two bucket sorts: first over the first elements of the pairs, then over the second elements.  Breaking the sorted list into chunks of pairs with the same first element (say $u$), gives the outgoing neighbors of $u$, in the second elements of the pairs, in a numerically sorted order. This takes $O(m)$ time.  Then find pairs of vertices to merge, if any exist: first, bucket sort vertices based on their out-degree. Vertices in different buckets cannot be merged by rule 1. For each bucket $b$ containing vertices with degree $d$ ($0 \leq d < n $), call MergeDetection$(b,d)$:

\begin{algorithm}[h]

    \SetKwFunction{FMain}{MergeDetection}
    \SetKwProg{Fn}{Function}{:}{}
    \Fn{\FMain{bucket $ a,i$}}{
        \eIf{$ i=0$}
            {return bucket $a$}
            {bucket sort vertices $v$ of $a$ based on S$[v][i]$

        \ForEach{newly created bucket $a'$ }
        {
       		 \FMain{$ a',i - 1$}
        }
        }
 }

\end{algorithm}

The vertices in each resulting bucket have the same outgoing neighbors and can be merged by rule 1. As each vertex with degree $d$ appears in one bucket in each of $d+1$ iterations, this sort takes $O(\sum_v(deg(v))) = O(m)$ time. Upon merging vertices $u$ and $v$, name the new vertex min$(u,v)$. 

To check rule 3, for each vertex $v$, compute $I(v)$: the intersection of the reachable sets of the incoming neighbors of $v$. This takes $O(mn)$ time. 
\iffalse each unlabeled edge $(u, v)$ participates in exactly one intersection, and each edge in an intersection contributes $O(n)$ time to the intersection operation.
\fi
Consider only those unlabeled edges $(u,v)$ that meet the preconditions of rule 3 concerning the existence of outgoing and incoming tasks of $u$ and $v$ respectively. Test whether the last point in rule 3 applies to edge $(u,v)$ by testing in $O(n)$ time whether all outgoing neighbors of $u$ are in $I(v)$. 

Computing the reachability matrix takes $O(mn)$ time, and using this matrix to check for rule 2 takes $O(m)$ time per iteration. Checking for rule 1 or 3 takes $O(mn)$ time per iteration. Further, the outer loop in Algorithm 1 runs at most $n$ times as it either merges two vertices or returns the output. This gives a total complexity of $O(mn^2)$ for our algorithm.

\section{Conclusion}
\label{sec:conclusion}
Our algorithm reduces the visual complexity of an activity-on-edge graph, making it easier to understand bottlenecks in a project. The algorithm repeatedly applies simple rules and therefore can be implemented easily. However, one can measure the complexity of a graph in other ways. One question for future work is whether one can minimize the number of edges in an AOE graph in polynomial time. Another question is
\iffalse
consider finding a graph with the most parsimonious structure with respect to edge crossings in all possible plane drawings. That is,
\fi
whether one can, in polynomial time, convert an AOE graph $G$ into a graph that (i) has the same critical paths as $G$, and (ii) has a plane drawing with fewer edge crossings than all other graphs satisfying (i).
\iffalse
In a similar vein, one could consider the problem, given an activity-on-edge graph $G$, of finding an equivalent activity-on-edge graph $H$ that has an upward planar drawing, if one exists.
\fi

%\bibliographystyle{plainurl} %lets use this for now even though prox will switch to splncs
\bibliography{ref}

\begin{thebibliography}{10}

\bibitem{ahr2010improving}
Sarmad Abbasi, Patrick Healy, and Aimal Rextin.
\newblock {Improving the running time of embedded upward planarity testing}.
\newblock {\em Information Processing Letters}, 110(7):274{--}278, 2010.
\newblock \href {http://dx.doi.org/10.1016/j.ipl.2010.02.004}
  {\path{doi:10.1016/j.ipl.2010.02.004}}.

\bibitem{agu1972transitive}
A.~V. Aho, M.~R. Garey, and J.~D. Ullman.
\newblock {The transitive reduction of a directed graph}.
\newblock {\em SIAM J. Comput.}, 1(2):131{--}137, 1972.
\newblock \href {http://dx.doi.org/10.1137/0201008}
  {\path{doi:10.1137/0201008}}.

\bibitem{arh1994nearcritical}
C.~Alexander, D.~Reese, and J.~Harden.
\newblock Near-critical path analysis of program activity graphs.
\newblock In {\em Proceedings of International Workshop on Modeling, Analysis
  and Simulation of Computer and Telecommunication Systems}, pages 308--317.
  {IEEE}, 1994.
\newblock \href {http://dx.doi.org/10.1109/mascot.1994.284406}
  {\path{doi:10.1109/mascot.1994.284406}}.

\bibitem{arditi2006selecting}
David Arditi and Thanat Pattanakitchamroon.
\newblock Selecting a delay analysis method in resolving construction claims.
\newblock {\em International J. Project Management}, 24(2):145--155, 2006.
\newblock \href {http://dx.doi.org/10.1016/j.ijproman.2005.08.005}
  {\path{doi:10.1016/j.ijproman.2005.08.005}}.

\bibitem{bm2001layered}
Oliver Bastert and Christian Matuszewski.
\newblock {Layered drawings of digraphs}.
\newblock In Michael Kaufmann and Dorothea Wagner, editors, {\em Drawing
  Graphs: Methods and Models}, volume 2025 of {\em Lecture Notes in Computer
  Science}, pages 87{--}120. Springer-Verlag, 2001.
\newblock \href {http://dx.doi.org/10.1007/3-540-44969-8_5}
  {\path{doi:10.1007/3-540-44969-8_5}}.

\bibitem{battista1998graph}
Giuseppe~Di Battista, Peter Eades, Roberto Tamassia, and Ioannis~G Tollis.
\newblock {\em Graph drawing: algorithms for the visualization of graphs}.
\newblock Prentice Hall PTR, 1998.

\bibitem{cg1986parallel}
Pranay Chaudhuri and Ratan~K. Ghosh.
\newblock Parallel algorithms for analyzing activity networks.
\newblock {\em {BIT}}, 26(4):418--429, 1986.
\newblock \href {http://dx.doi.org/10.1007/bf01935049}
  {\path{doi:10.1007/bf01935049}}.

\bibitem{bt1992area}
Giuseppe Di~Battista, Roberto Tamassia, and Ioannis~G. Tollis.
\newblock {Area requirement and symmetry display of planar upward drawings}.
\newblock {\em Discrete and Computational Geometry}, 7(4):381{--}401, 1992.
\newblock \href {http://dx.doi.org/10.1007/BF02187850}
  {\path{doi:10.1007/BF02187850}}.

\bibitem{es2013confluent}
David Eppstein and Joseph~A. Simons.
\newblock {Confluent Hasse diagrams}.
\newblock {\em J. Graph Algorithms Appl.}, 17(7):689{--}710, 2013.
\newblock \href {http://dx.doi.org/10.7155/jgaa.00312}
  {\path{doi:10.7155/jgaa.00312}}.

\bibitem{gt1995upward}
Ashim Garg and Roberto Tamassia.
\newblock {Upward planarity testing}.
\newblock {\em Order}, 12(2):109{--}133, 1995.
\newblock \href {http://dx.doi.org/10.1007/BF01108622}
  {\path{doi:10.1007/BF01108622}}.

\bibitem{gt2001computational}
Ashim Garg and Roberto Tamassia.
\newblock {On the computational complexity of upward and rectilinear planarity
  testing}.
\newblock {\em SIAM J. Comput.}, 31(2):601{--}625, 2001.
\newblock \href {http://dx.doi.org/10.1137/S0097539794277123}
  {\path{doi:10.1137/S0097539794277123}}.

\bibitem{hkl2004characterization}
Patrick Healy, Ago Kuusik, and Sebastian Leipert.
\newblock A characterization of level planar graphs.
\newblock {\em Discrete Math.}, 280(1-3):51--63, 2004.
\newblock \href {http://dx.doi.org/10.1016/j.disc.2003.02.001}
  {\path{doi:10.1016/j.disc.2003.02.001}}.

\bibitem{hn2014hierarchical}
Patrick Healy and Nikola~S. Nikolov.
\newblock {Hierarchical Graph Drawing}.
\newblock In Roberto Tamassia, editor, {\em Handbook of Graph Drawing and
  Visualization}, pages 409{--}453. CRC Press, 2014.

\bibitem{householder1990owns}
Jerry~L. Householder and Hulan~E. Rutland.
\newblock Who owns float?
\newblock {\em J. Construction Engineering and Management}, 116(1):130--133,
  1990.
\newblock \href {http://dx.doi.org/10.1061/(ASCE)0733-9364(1990)116:1(130)}
  {\path{doi:10.1061/(ASCE)0733-9364(1990)116:1(130)}}.

\bibitem{hd2011project}
Wei Huang and Lixin Ding.
\newblock Project-scheduling problem with random time-dependent activity
  duration times.
\newblock {\em {IEEE} Transactions on Engineering Management}, 58(2):377--387,
  2011.
\newblock \href {http://dx.doi.org/10.1109/tem.2010.2063707}
  {\path{doi:10.1109/tem.2010.2063707}}.

\bibitem{junger1999level}
Michael J{\"{u}}nger and Sebastian Leipert.
\newblock {Level Planar Embedding in Linear Time}.
\newblock In Jan Kratochv{\i}l, editor, {\em Graph Drawing: 7th International
  Symposium}, volume 1731 of {\em Lecture Notes in Computer Science}, pages
  72{--}81. Springer-Verlag, 1999.
\newblock \href {http://dx.doi.org/10.1007/3-540-46648-7_7}
  {\path{doi:10.1007/3-540-46648-7_7}}.

\bibitem{kulkarni1984compact}
Vidyadhar~G. Kulkarni.
\newblock A compact hash function for paths in {PERT} networks.
\newblock {\em Operations Research Letters}, 3(3):137--140, 1984.
\newblock \href {http://dx.doi.org/10.1016/0167-6377(84)90005-1}
  {\path{doi:10.1016/0167-6377(84)90005-1}}.

\bibitem{malcolm1959application}
Donald~G. Malcolm, John~H. Roseboom, Charles~E. Clark, and Willard Fazar.
\newblock Application of a technique for research and development program
  evaluation.
\newblock {\em Operations Research}, 7(5):646--669, 1959.
\newblock \href {http://dx.doi.org/10.1287/opre.7.5.646}
  {\path{doi:10.1287/opre.7.5.646}}.

\bibitem{sugiyama1981methods}
Kozo Sugiyama, Shojiro Tagawa, and Mitsuhiko Toda.
\newblock Methods for visual understanding of hierarchical system structures.
\newblock {\em IEEE Trans. Systems, Man, and Cybernetics}, 11(2):109--125,
  1981.
\newblock \href {http://dx.doi.org/10.1109/TSMC.1981.4308636}
  {\path{doi:10.1109/TSMC.1981.4308636}}.

\bibitem{xu2010automatic}
Zhenming Xu.
\newblock Automatic layout of information in the {AOE} network.
\newblock In {\em 2010 International Conference on Mechanic Automation and
  Control Engineering}. {IEEE}, June 2010.
\newblock \href {http://dx.doi.org/10.1109/mace.2010.5536213}
  {\path{doi:10.1109/mace.2010.5536213}}.

\end{thebibliography}

\newpage
\appendix
\section{Appendix}

\subsection{Proof of \autoref{lem:endstartvertices}}

\label{app:endstartvertices}
\endstartvertices*
\begin{proof}

Suppose for a contradiction that $u = \en_\alg(T) \neq v = \stv_\alg(T')$, but $\en_\opt(T) = \stv_\opt(T')$. Since $T \rsag{\opt} T'$, then by Lemma~\ref{lem:cpimpliespo}, $T \rsag{\alg} T'$, so $u$ has a path $P$ to $v$. Consider the following possible cases for path $P$:

\begin{enumerate}
    \item There exists a task $A$ in $P$.
    \label{case:TinP}
    \item $P$ only consists of unlabeled edges. Consider the cases for any unlabeled edge $(u',v')$ in $P$:
    \label{case:noTinP}
        \begin{enumerate}
        % case 2a
        \item There exist incident tasks $S$ and $S'$, pointing away from and toward $u'$ and $v'$, respectively.  
        \label{case:TT}
        % case 2b
        \item There exists an incident unlabeled edge $(u',w')$ pointing away from $u'$ and an incident task $S'$ pointing toward $v'$. 
        \label{case:UT}
        % case 2c
        \item There exists an incident task $S$ pointing away from $u'$ and  an  incident unlabeled edge $(w',v')$ pointing toward $v'$.
        \label{case:TU}
        % case 2d
        \item There exists an  incident unlabeled edge $(u',w')$ pointing away from $u'$ and an incident unlabeled edge $(x',v')$ pointing toward $v'$. Vertices $u'$ and $v'$ have no outgoing or incoming task edges, respectively.
        \label{case:UU}
        \end{enumerate}
\end{enumerate}

These cases are exhaustive as path $P$ either has a task or it is a sequence of unlabeled edges. Further, for case~\ref{case:noTinP}, suppose for an unlabeled edge $(u',v')$, none of the subcases of \ref{case:TT}, \ref{case:UT}, \ref{case:TU} and \ref{case:UU} holds. Then, by rule~\ref{rule:AllHasPath}, one can merge vertices $u'$ and $v'$; this contradicts the definition of \alg. 

\begin{figure}[t]
\center
  \includegraphics[scale=0.8]{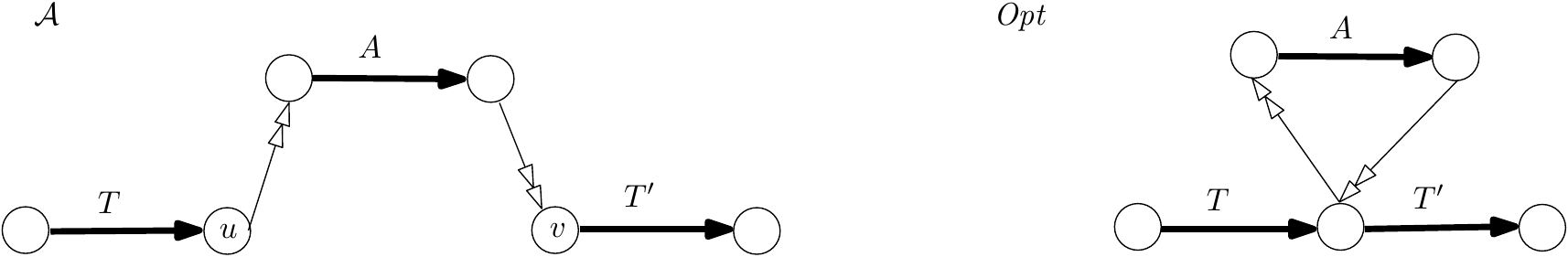}
  \caption{\autoref{lem:endstartvertices}, case~\ref{case:TinP}.}
  \label{fig:endstartcase1}

  \includegraphics[scale=0.8]{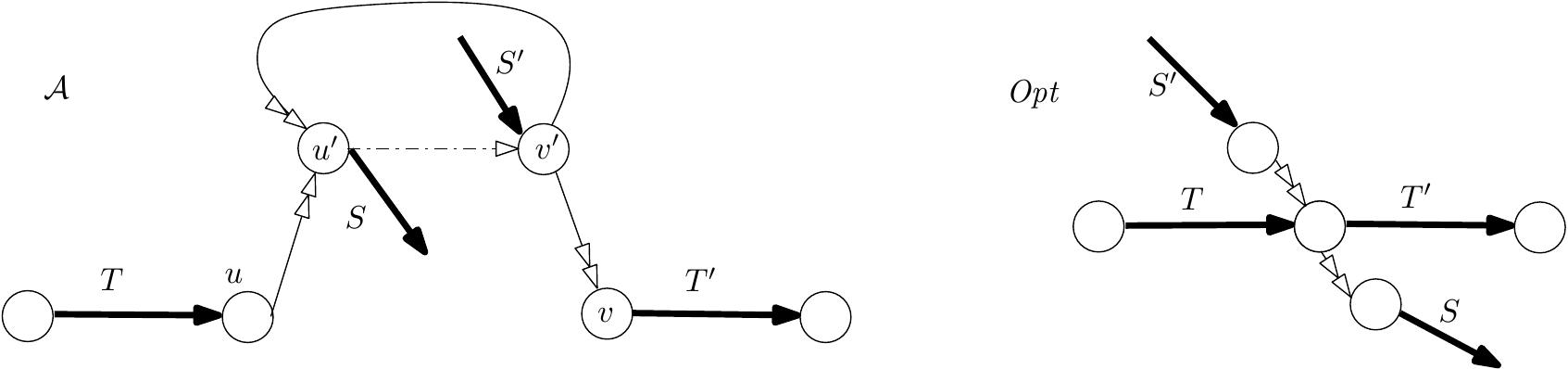}
  \caption{\autoref{lem:endstartvertices}, case~\ref{case:TT}.}
  \label{fig:endstartcase2a}
\end{figure}

\begin{figure}
\center
  \includegraphics[scale=0.8]{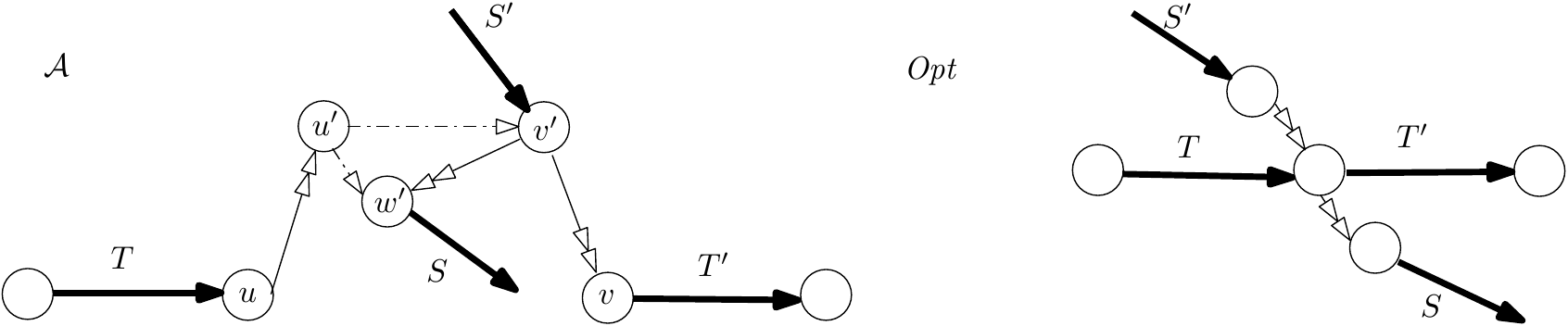}
  \caption{\autoref{lem:endstartvertices}, case~\ref{case:UT}.}
  \label{fig:endstartcase2b}

  \includegraphics[scale=0.8]{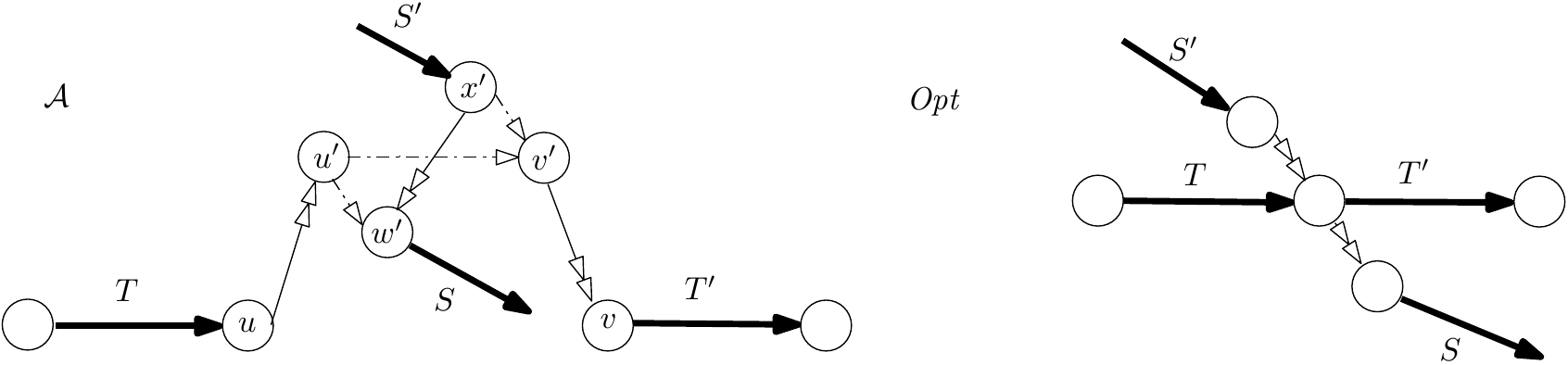}
  \caption{\autoref{lem:endstartvertices}, case~\ref{case:UU}.}
  \label{fig:endstartcase2d}
\end{figure}

In case~\ref{case:TinP}, shown in \autoref{fig:endstartcase1}, we have $T \rsag{\alg} A \rsag{\alg} T'$ so by \autoref{lem:cpimpliespo}, $T \rsag{\opt} A \rsag{\opt} T'$. Therefore, there is a path in \opt (through $A$) from $\en_\opt(T)$ to $\stv_\opt(T')$. Since $\en_\opt(T) = \stv_\opt(T')$, this path creates a cycle in \opt. However, \opt is an AOE graph and is therefore acyclic by \autoref{def:CPG}.

In case~\ref{case:TT}, shown in \autoref{fig:endstartcase2a}, $S' \rsag{\alg} T'$, so by \autoref{lem:cpimpliespo}, $S' \rsag{\opt} T'$, i.e. there is a path from $S'$ to $\stv_\opt(T')$. Similarly, $T \rsag{\alg} T'$, so there is a path from $\en_\opt(T)$ to $S$. Since $\stv_\opt(T') = \en_\opt(T)$, this implies $S' \rsag{\opt} S$, so by \autoref{lem:cpimpliespo}, $S' \rsag{\alg} S$. Therefore, there is a path in \alg from $\en_\alg(S')$ to $\stv_\alg(S)$. This path creates a cycle in \alg and contradicts \autoref{cor:nocycle}.

%2b
In case~\ref{case:UT}, shown in \autoref{fig:endstartcase2b}, by \autoref{lem:taskdensity}, there exists a task $S$ where $w' = \stv_\alg(S)$. Since we have $T \rsag{\alg} S $ and $S' \rsag{\alg} T'$ and $\en_\opt(T) = \stv_\opt(T')$, by \autoref{lem:cpimpliespo}, we have  $S' \rsag{\opt} S $. Therefore, there is a path in \alg from $v' = \en_\alg(S')$ to $w' = \stv_\alg(S)$. This path either creates a cycle between $u'$ and $v'$, contradicting \autoref{cor:nocycle} or by rule~\ref{rule:HasPath}, one can remove edge $(u',w')$, which is a contradiction by the definition of \alg.

Case~\ref{case:TU} is almost identical to case~\ref{case:UT}, and again leads to the existence of a path from $S'$ to $S$ (similarly defined), resulting in either a cycle or an application of rule~\ref{rule:HasPath} .

In case~\ref{case:UU}, shown in \autoref{fig:endstartcase2d}, by \autoref{lem:taskdensity}, there exist task edges $S$ and $S'$ where $w' = \stv_\alg(S)$ and $x' = \en_\alg(S')$, and unlabeled edges $(x', v')$ and $(u', w')$. We have $S' \rsag{\opt} S$, then by \autoref{lem:cpimpliespo}, $S' \rsag{\alg} S$. Therefore, there is a path in \alg from $\en_\alg(S')$ to $\stv_\alg(S)$. This path either creates a cycle between $u'$ and $v'$, contradicting \autoref{cor:nocycle} or by rule~\ref{rule:AllHasPath}, one can merge $u'$ and $v'$ in \alg, which is a contradiction by the definition of \alg.

Thus if $\en_\alg(T) \neq \stv_\alg(T')$, then $\en_\opt(T) \neq \stv_\opt(T')$. 

\end{proof}

\end{document}